\documentclass[12pt]{article}
\usepackage[margin=1in]{geometry} 
\usepackage{amsmath}
\usepackage{tcolorbox}
\usepackage{amssymb}
\usepackage{amsthm}
\usepackage{mathtools}
\usepackage{lastpage}
\usepackage{accents}
\usepackage{systeme}
\usepackage{tikz}
\usepackage{enumerate}
\usepackage{tikz-cd}
\usepackage{enumitem}
\usepackage{bbm}
\usepackage{hyperref}
\usepackage[capitalise]{cleveref}

\newtheorem{theorem}{Theorem}[section]
\newtheorem{proposition}{Proposition}[section]
\newtheorem{lemma}{Lemma}[section]

\newtheorem{claim}{Claim}[section]

\newtheorem{definition}{Definition}[section]

\theoremstyle{definition}

\newtheorem*{case*}{Case}

\newcommand{\R}{\mathbb{R}}
\newcommand{\Rg}{\mathbb{R}_{\ge 0}}

\newcommand{\bra}[1]{\left[#1\right]}

\newcommand{\set}[1]{\left\{ #1 \right\}}
\newcommand{\abs}[1]{\left|#1\right|}

\newcommand{\B}[1]{\boldsymbol{#1}}

\newcommand{\mc}{\mathcal}

\makeatletter
\let\@fnsymbol\@arabic
\makeatother

\title{Resonance: Transaction Fees for Heterogeneous Computation}

\author{Maryam Bahrani$^*$\thanks{Ritual. Email: \texttt{maryam@ritual.net}.} \and Naveen Durvasula$^*$\thanks{Ritual \& Columbia University. Email: \texttt{naveen@ritual.net}.}}
\date{}
\allowdisplaybreaks

\begin{document}

\maketitle

\def\thefootnote{*}\footnotetext{The authors contributed equally to this work.}
\def\thefootnote{\arabic{footnote}}

\begin{abstract}
Blockchain networks are facing increasingly heterogeneous computational demands, and in response, protocol designers have started building specialized infrastructure to supply that demand.
This paper introduces \textit{Resonance}: a new kind of transaction fee mechanism for the general two-sided market setting (with \emph{users} on one side and \emph{nodes} on the other), where both sides of the market exhibit a high degree of heterogeneity.
We allow users submitting transactions to have arbitrary valuations for inclusion, nodes responsible for executing transactions to incur arbitrary costs for running any bundle of transactions, and further allow for arbitrary additional constraints on what allocations are valid. 
These constraints can, for example, be used to prevent state conflicts by requiring transactions that utilize the same part of the network's state to not be executed in parallel. They also enable support for new transaction types, such as transactions that require multiple nodes for execution (e.g. to run multi-party computation for better transaction privacy).
    
Resonance's design utilizes competition among sophisticated \emph{brokers} to find individualized prices for each transaction and node. We show that at pure Nash equilibria, Resonance finds an efficient outcome and minimizes the need for strategization by users and nodes. It is also budget-balanced, individually rational for all parties, and computationally tractable. 
\end{abstract}

\section{Introduction}
Public blockchain networks utilize decentralized consensus protocols to provide a ``trusted'' computing environment; ideally, these protocols provide users practical assurances that transactions they submit to the network will be faithfully executed to modify a consistent and publicly-accessible network state. As more applications for this type of computing are discovered, users' computational demands on blockchain networks have become increasingly \textit{heterogeneous}. Initially, with Bitcoin, computational demands were restricted to simple arithmetic operations as users of the network were primarily interested in making financial transactions. Today, users make increasingly sophisticated computational demands in the form of arbitrary smart contracts. The magnitude and scope of demand for computation continues to expand as more specialized on-chain services are deployed and get traction. \cite{polkadot, cosmos, raas}. 

In response to this increasing heterogeneity in use cases, protocol designers have sought to develop specialized computational infrastructure to supply this demand. A concrete and recent example of such specialization on the Ethereum blockchain is EIP-4844~\cite{4844}. This proposed change to Ethereum enables a specialized form of short-term storage called ``blobs''. These improvements were designed in response to the computational demands that rollup networks and other services now have (\cite{scroll-costs,zksync,matic}). Numerous sovereign blockchains and rollups themselves also specialize in executing certain types of computational demand (e.g. proving services such as Lagrange\footnote{\url{https://docs.lagrange.dev/zk-coprocessor/overview}}, Gevulot\footnote{\url{https://gevulot.com/}}, and SP1\footnote{\url{https://blog.succinct.xyz/introducing-sp1/}}, ML and AI training services \cite{bittensor,gensyn-litepaper,modulus}, FHE \cite{fhenix-whitepaper, zama-whitepaper}, oracle networks \cite{chainlink,gelato}). As a result, there is increasing heterogeneity not only in demand but also supply of trusted computation.

In order to efficiently match computational supply with demand, blockchain networks use market mechanisms that set fees and rewards for users and computational suppliers respectively. Major blockchains including Ethereum measure the resource usage of a given transaction using a single numeraire, and set prices for transaction inclusion by setting a unit price for the numeraire. These fee mechanisms have been studied extensively in recent work, and their implementations (e.g. EIP-1559) satisfy desirable properties.\cite{tfm-1559,tfm-foundations,tfm-collusion} In order to efficiently match demand and supply in the presence of heterogeneity, however, these market mechanisms must be adapted. The Ethereum blockchain, for example, has recently added a second resource numeraire to price the usage of blobs as part of EIP-4844 \cite{4844}. This general idea of having multiple resource numeraires, or ``dimensions,'' across which the heterogeneity in computational demand is expressed, and subsequently setting a unit price for usage across each numeraire to clear demand is known as a \textit{multi-dimensional fee market}. Recent work has studied their properties \cite{multidim1,multidim2,chaos,multidim3,multidim4}. Services beyond Ethereum and other layer one blockchains face even more complexity due to greater levels of heterogeneity in demand and supply. Prover networks \cite{MD-for-prover-markets,lagrange-paper} are a notable example, where users have heterogeneous values for proof capabilities of different suppliers, and suppliers have heterogeneous costs and constraints for each transaction due to variability in hardware quality and specialization as well as the ability to parallelize. Another example is Ritual\footnote{\url{https://ritual.net/}}, which seeks to enable users demanding highly heterogeneous and expensive AI and ML compute to execute transactions over a highly heterogeneous collection of compute nodes with different resources.
This paper describes and analyzes a market design problem that generalizes all of these instances, with full heterogeneity on both sides of the market.

\subsection{Summary of Results}

Our principal contribution in this paper is the description and analysis of the Resonance mechanism, which succeeds in an extremely general setting, with fully heterogeneous supply and demand and arbitrary additional constraints on the allocation.

To describe our model at a high level, we consider a two-sided marketplace, with users demanding computation on one side and nodes supplying computational resources on the other. Each user submits a transaction, and has an arbitrary (private) value for that transaction getting executed. Each node has an arbitrary (private) cost function for executing any bundle of transactions. We define an \emph{allocation} as an assignment of each transaction to a set of nodes, and allow for an arbitrary validity constraint on the set of ``allowable'' allocations. These constraints can for example ensure that transactions that touch the same part of the network state are run by the same node, or state that nodes are not assigned transactions that exceed their capacity.

This setting is very general, and in particular captures previous work on fee mechanisms. For example, if there is only one node, which has zero cost for running transactions, and the validity constraint can be expressed as a capacity constraint on a single numeraire (i.e. an allocation is valid if the total gas consumption of allocated transactions does not exceed the maximum block size), we recover the setting of Ethereum fee markets with a single dimension. As another example, if again there is a single numeraire (denoting CPU cycles), but more than one node, and the cost function of each node $n$ is a node-specific constant $x_n$ times the total units of the numeraire assigned to that node, then we recover the setting of prover markets as studied in \cite{MD-for-prover-markets,lagrange-paper}.

A fee mechanism in this setting must determine an allocation of transactions to nodes, and prices paid by each transactions as well as paid to each node. Users and nodes have quasilinear utilities, meaning users maximize their value minus payment, and nodes maximize their reward minus cost. The goal is to design a mechanism that satisfies several desirable properties (stated informally):
\begin{enumerate}
    \item \emph{Budget-Balance}: The mechanism should not require subsidization to run.
    \item \emph{Individual Rationality:} No one is worse off by participating in the mechanism than abstaining.
    \item \emph{Efficiency:} The mechanism must maximize \emph{surplus}, defined as the sum of utilities of users and nodes.
    \item \emph{Incentive-Compatibility:} The mechanism does not require complex strategization by users and nodes.
    \item \emph{Tractability}: The mechanism is computationally lightweight and feasible to implement on-chain.
\end{enumerate}

We motivate the design of the Resonance Mechanism by first examining whether existing solutions can achieve these properties. In \cref{sec:mdfm} we perform a formal analysis of multi-dimensional fee markets (such as multi-dimensional EIP-1559) with an emphasis on efficiency in particular. Our results imply that in a heterogeneous setting, selecting allocations and setting prices are both complicated in their own right, \emph{even if one has access to an optimal oracle for the other}.

 Furthermore, pricing is challenging because higher heterogeneity necessitates a higher number of prices. \cref{thm:wo} quantifies this intuition, showing that even with an optimal allocation oracle, surplus maximization (i.e. efficiency) requires a separate price for each transaction. This limits the applicability of the currently popular approach of using an on-chain dynamic price oracle that updates according to historical signals about supply and demand. Implementing a pricing oracle at that level of granularity on-chain is infeasible both computationally and information-theoretically.

 Our high-level solution to these challenges is to off-load both of these complex tasks to a third party of agents that we call \textit{brokers}. These sophisticated agents are responsible for finding allocations and prices, and in exchange can extract part of the generated surplus from the allocation they find. We formally define the five properties for a setting with brokers as well as users and nodes, and show that the Resonance Mechanism satisfies all properties: It is (a) budget-balanced, (b) individually rational for users, nodes, and brokers, (c) surplus-maximizing in equilibrium, (d) incentive-compatible for users and nodes (conditional on brokers playing certain natural equilibrium strategies), and (e) computationally tractable.


The mechanism itself is quite simple. Each broker submits an allocation and individualized payments for all users and nodes. The auctioneer collects the broker submissions, as well as valuations and costs from users and nodes, and finds the surplus maximizing broker submission (as computed given the elicited valuations and costs). If that broker submission satisfies individual rationality for all users and nodes (\textit{i.e.}, results in non-negative utility for them given their elicited valuations/costs), the mechanism outputs the broker's submission, and pays out any leftover cash after settling payments to the chosen broker.

The analysis of the mechanism's incentive properties in particular is subtle. Since strong notions such as dominant-strategy incentive-compatibility (DSIC) are impossible to achieve in this setting \cite{myerson1981optimal}, we introduce a new version of incentive-compatibility that our mechanism satisfies, which might be of independent interest. This definition (\cref{def:dsic-barring-s}) is appropriate for a setting where there are two types of agents, ordinary and sophisticated, and the primary objective is to minimize the need for strategization by ordinary agents. The best-response conditions of the new definition accordingly only checks the best-response conditions of each ordinary agent with respect to alternative actions of \emph{other ordinary agents}, while keeping keeping the strategies of sophisticated agents fixed. This relaxation of incentive-compatibility can provide an assurance to an ordinary agent that might be considering reasoning about strategies of other ordinary agents that such strategization will not be profitable. However, it provides no guarantees with respect to alternative strategies of sophisticated agents. Assuming ordinary agents have similar levels of rationality and resources, it is plausible that they would entertain the former category of strategization but not the latter.

We conclude the paper by listing some practical considerations when deploying Resonance. We give a high level overview of potential extensions to the mechanism that could address these considerations, and defer a more rigorous treatment to future work.

\section{The Heterogeneous Setting}\label{sec:model}

\paragraph{Agents and Their Types.}
In the heterogeneous setting, there are two types of agents, users who demand computation, and nodes who supply it. Each user submits a transaction $t$, which we will use to refer to both the transaction and the user who submitted it.

Let $T$ be the set of transactions (and users), $N$ the set of nodes, and $A:=T\cup N$ the set of all agents. For each transaction $t \in T$, there is a valuation $v_t\in\mathbb{R}_{\ge0}$ that denotes the maximum amount of money that the user $t$ is willing to pay to have $t$ be executed. For every node $n \in N$, there is a cost function $c_n(X):2^T\rightarrow\mathbb{R}_{\ge 0}$ that denotes the amount of money node $n$ must be paid in order to execute a batch of transactions $X \subseteq T$. We require that $c_n(\emptyset)=0$, but make no other assumptions about the cost functions. This degree of generality allows nodes to have custom costs depending on their hardware, what they can run in parallel, and how they choose to schedule and run the batch of jobs that are assigned to them. Initially, $v_t$ and $c_n$ are only privately known by the respective users and nodes. We refer to these quantities as the agents' \emph{types}, with $\theta_a$ denoting the type of agent $a\in A$ defined as $v_a$ if $a\in T$ and $c_a$ if $a\in N$. Let $\Theta_a$ be the set of all possible types for agent $a$, and $\Theta=\bigtimes_{a\in A}\Theta_a$ the space of type profiles.

\paragraph{Allocations.}
An \emph{allocation} $\alpha:T\rightarrow 2^N$ assigns each transaction to a (possibly empty) set of nodes. The \emph{inverse} of an allocation $\alpha$ is defined as $\alpha^{-1}(n) := \{t \in T \mid n \in \alpha(t)\}$, which specifies the transactions that a node $n$ executes under the allocation $\alpha$.

Given an allocation $\alpha$, the \emph{transactions of $\alpha$}, denoted by $T(\alpha)$, are the set of transactions assigned by $\alpha$ to a (non-empty) set of nodes, that is $T(\alpha) := \bigcup_{n\in N}\alpha^{-1}(n)$. Similarly, the \emph{nodes of $\alpha$}, denoted by $N(\alpha)$, are the set of nodes that $\alpha$ assigns at least one transaction to, that is $N(\alpha) := \bigcup_{t\in T}\alpha(t)$.


\paragraph{Validity and Conflicts.}

Only some allocations may be \emph{valid}. In the context of Ritual Chain, for example, two transactions that touch the same state should be run by the same node to avoid causing state conflicts. Thus, we would call any allocation that maps such transactions to different nodes invalid. Alternatively, we might have protocol-wide upper bounds on the consumption of each resource per batch, which we can enforce by labelling allocations that exceed these bounds as invalid.

We denote by $V$ the set of all valid allocations. We make no assumptions about $V$ other than the fact that it includes the empty allocation (which assigns all transactions to the empty set). This degree of generality means our mechanism can capture arbitrarily complex constraints, including state conflicts and congestion bounds described in the previous paragraph.


\paragraph{Routings.}

A \emph{transaction payment rule} $\pi: T \to \mathbb R_{\ge 0}$ determines how much each transaction should pay for execution. A \emph{node payment rule} $\phi: N \to \mathbb R_{\ge 0}$ determines how much each node should be paid for the execution of transactions assigned to it.

A \emph{routing} $(\alpha, \pi, \phi)$ specifies a valid allocation $\alpha \in V$, and the user and node payment rules. Note that that the definition of routings requires that the allocation is valid. 
We denote by $\mc{R}$ the set of all routings.

The \emph{margin} of a routing $(\alpha, \pi, \phi)$ measures the net cash flow resulting from it, and is given by $\sum_{t\in T}\pi(t)-\sum_{n\in N}\phi(n).$ We denote by $\gamma$ the function that takes as input a routing and outputs its margin.
A routing is called \emph{budget-balanced} if its margin is non-negative.

We denote by $R_{\emptyset}$ the \emph{empty routing}, which has an allocation rule mapping all transactions to the empty set, and payment rules that are equal to zero for all transactions and nodes.

\paragraph{Utilities.}
Under a routing $R=(\alpha, \pi, \phi)$, a user who submits a transaction $t \in T$ and has valuation $v_t$ obtains utility
\[
u(t,R\mid v_t) := \underbrace{\mathbbm{1}[\alpha(t) \ne \emptyset] \cdot v_t}_{\text{Valuation for inclusion in the allocation}} - \underbrace{\pi(t)}_{\text{Amount paid for execution}}
\]
and a node $n\in N$ with cost function $c_n$ obtains utility
\[
u(n, R \mid c_n) := \underbrace{\phi( n)}_{\text{Amount paid to node}} - \underbrace{c_n(\alpha^{-1}(n))}_{\text{Cost incurred by node to execute allocated transactions}}
\]


\paragraph{Efficient Routings.} 
The \emph{surplus} of a routing $R$ is the sum of utilities of transactions and nodes under that routing,
\[
\mathcal{S}(R \mid \B\theta) := \sum_{a \in A} u(a,R\mid \theta_a).
\]
A routing is called \emph{efficient} if it maximizes surplus.
Sometimes we drop the $\B\theta$ when invoking $\mathcal S$ when its values are clear from context.





\section{Desired Properties in a Mechanism}
\label{sec:broker}




It turns out in order to obtain good welfare guarantees while maintaining a simple user experience in a heterogeneous setting, transactions must be priced individually (we refer the reader to \cref{sec:mdfm} for a formal argument justifying this claim).
The natural question is how to discover prices at such a high level of granularity. In particular, existing solutions that implement an on-chain price oracle based on historical demand run into information theoretic and computational issues if extended to provide individualized pricing. This is because the amount of data necessary to learn a high-dimensional price vector is very large, and the aggregation mechanism necessary for an accurate estimator needs to be a sophisticated algorithm that is costly to run on-chain. Can we design a tractable mechanism that works in the fully heterogeneous setting?

The following are properties we desire in a routing mechanism (stated informally here and formalized later).
\begin{enumerate}
    \item \emph{Budget-Balance}: The mechanism should not require subsidization to run.
    \item \emph{Individual Rationality:} No one is worse off by participating in the mechanism than abstaining.
    \item \emph{Efficiency:} The mechanism finds a surplus-maximizing routing. 
    \item \emph{Incentive-Compatibility:} The mechanism is ex-post incentive compatible, meaning truth-telling by users and nodes is a Nash equilibrium.
    \item \emph{Tractability}: The mechanism is computationally lightweight and feasible to implement on-chain.
\end{enumerate}


A simple special case of our setting is the problem of bilateral trade, where there is one transaction and one node, and the transaction has value $v$ for being assigned while the node has cost $c$ for executing the transaction. When these values are private, it is well-known that the first four properties cannot be simultaneously satisfied. \cite{green-laffont-1979,MS} On the other hand, if the mechanism has access to these values, the following simple design satisfies all four properties: Compute an efficient routing, and present it as a take-it-or-leave-it offer to every transaction and node, with trade happening only if all parties take the offer. That is, without private information, surplus-maximization is simply an optimization problem. However, even with a perfect price oracle, the computational task of finding an efficient \emph{allocation} with arbitrary feasibility constraints can be hard to even approximate, failing the \emph{tractability} requirement. This is in addition to the tractability issues of implementing the price oracle itself.

Motivated by these challenges, we propose a mechanism that delegates the complicated tasks of estimating prices and optimal allocations to a third set of agents, called \emph{brokers}. Brokers specialize in finding efficient routings, and are compensated for their efforts by earning some of the generated surplus, called \emph{margin}. As a result, the incentives of brokers are aligned with finding the surplus-maximizing allocation, and competition between brokers results in low extracted margins. This effectively implements an allocation and price oracle, which the mechanism uses to achieve efficiency as well as incentive-compatibility for transactions and nodes. While mechanism designers typically strive to minimize the need for participant sophistication, here we only aim to minimize that need for a subset of participants (users and nodes), while relying on sophistication of other participants (brokers) to ease the computational cost of running the mechanism itself.
We will next describe the Resonance mechanism, and show that it satisfies all five properties. 




\section{The Resonance Mechanism: Description} \label{sec:broker-spec}


Resonance is formalized as the following game-form between transactions $T$, nodes $N$, and a set of brokers $B$. We will denote by $A$ the set $A$ of all transactions and nodes.



\paragraph{Information Structure.}
As before, users and nodes have private \emph{types} representing their payoff functions, with $\theta_a\in\Theta_a$ denoting the type of agent $a\in A$ defined as $v_a$ if $a\in T$ and $c_a$ if $a\in N$. In this paper, we consider a setting where every broker knows the exact types of every transaction and node (and transactions and nodes are aware of this fact).\footnote{See \emph{Broker Specialization} in \cref{sec:conclusion} for a discussion of when brokers have incomplete information.} Brokers do not have types.

\paragraph{Action Spaces.}
The mechanism is direct for transactions and nodes, meaning their action spaces are the same as their type spaces. Each agent $a\in A$ submits a type $\theta'_a\in \Theta_a$ to the mechanism (which may not be the same as $a$'s true type $\theta_a$).

Each broker $b$ submits a \emph{proposal} to the mechanism, which consists of a routing $R_b$. We will denote by $\Sigma_b$ the action space of broker $b$, and by $\Sigma=\bigtimes_{b\in B}\Sigma_b$ the action space of all brokers.

\paragraph{The Mechanism.}

Resonance is a deterministic map that takes as input an an action profile $(\B\theta',\B\sigma) \in \Theta\times\Sigma$ and outputs a routing $R$ as well as a winning broker $b$.

Let $b^*$ be the the broker who submitted the highest-surplus budget-balanced proposal given the submissions of transactions and nodes, that is
\[ b^* := \arg\max_{b\in B}\set{\mc S(R_b\mid \B \theta') \mid \text{ $R_b$ budget-balanced}},\]
breaking ties according to an arbitrary fixed ordering over the brokers (that does not depend on the actions of any participants). Let $R^*$ be the routing submitted by $b^*$.

If $R^*$ results in negative utility for any transaction or node assuming true types $\B\theta'$, (\textit{i.e.} if there exists some  some $a\in A$ such that $u(a,R^*\mid \theta_a)< 0$), then the mechanism outputs the empty routing. Otherwise, the mechanism outputs $(R^*,b^*)$.

\paragraph{Payments and Utilities.}
Payments to users and nodes are settled according to $R^*$. Payments to brokers are set to zero if the mechanism outputs the empty routing. Otherwise, payments are set to zero for all $b\not= b^*$, and to the margin of $R^*$ for $b^*$.

Users and nodes are quasilinear, and brokers are revenue-maximizing. The utility of a user or node $a$ with type $\theta_a$ under action profile $(\B\theta',\B\sigma) \in \Theta\times\Sigma$ is denoted by $u(a,\B\theta',\B\sigma\mid \theta_a)$ and is defined as their utility under the output routing when everyone plays $(\B\theta',\B\sigma)$. The utility of a broker $b$ under $(\B\theta',\B\sigma)$ is denoted by $u(b,\B\theta',\B\sigma)$ and is equal to the payment it receive from the mechanism when everyone plays $(\B\theta',\B\sigma)$.
\section{The Resonance Mechanism: Analysis}


\subsection{Budget-Balance}
We say a mechanism is \emph{weakly budget-balanced} if its net inflow of cash is weakly positive, with equality in the case of \emph{strong budget-balance}. The following shows that Resonance is strongly budget-balanced.
\begin{lemma}\label{lem:BB}
    Fix an action profile $(\B\theta',\B\sigma)$, and let $R^*=(\alpha,\pi,\phi)$ be the routing output by the mechanism and $b^*$ the winning broker. We have
\[
    \underbrace{\sum_{t \in T}\pi(t)}_{\text{Total inflow of cash }} =
    \underbrace{\sum_{n \in N}\phi(n) + \gamma(R^*)}_{\text{Total outflow of cash}}.
\]
\end{lemma}
\begin{proof}
    This follows immediately from the definition of margin.
\end{proof}

\subsection{Individual Rationality}
Individual rationality is a participation constraint guaranteeing that agents are no worse off for electing to participate in the mechanism. A mechanism is \emph{ex post individually rational} for a player if, knowing its own type and the types of everyone else, the player always chooses to participate in the mechanism. The following lemma shows that Resonance is individually rational for all participants.

\begin{lemma}[Individual rationality for transactions and nodes]\label{lem:IR-agent}
For all transactions and nodes $a$, all types $\theta_a\in\Theta_a$, and all action profiles $(\B{\theta}'_{-a},\B\sigma)$, truth telling results in non-negative utility for $a$:
\[
u(a,\theta_a,\B{\theta}'_{-a},\B\sigma\mid \theta_a) \geq 0.
\]
\end{lemma}
\begin{proof}
    A transaction or node can guarantee a utility of at least zero by truthfully reporting its type. This is because the mechanism only outputs routings that result non-negative utility for transactions and nodes under the reported types; if no such proposal exists, the mechanism outputs the null allocation and zero payment rules, under which all transactions and nodes have utility zero.
\end{proof}

\begin{lemma}[Individual rationality for brokers]\label{lem:IR-broker}
For all brokers $b$ and all action profiles $(\B{\theta}',\B\sigma_{-b})$, submitting the empty routing guarantees non-negative utility for $b$:
\[
u(b,\B{\theta}',R_\emptyset,\B\sigma_{-b})=0.
\]
\end{lemma}
\begin{proof}
    If broker $b$ submits $R_\emptyset$ and wins, it receives a payment of $\gamma(R_\emptyset)=0$. If $b$ does not win, it receives a payment of zero. The utility of $b$ is equal to its payment, so both cases result in zero utility for $b$.
\end{proof}

\subsection{Efficiency}
\paragraph{Welfare.} A common metric for the efficiency of an outcome is the sum of utilities of all participants.

Formally, the \emph{welfare} $\mathcal W$ of an action profile $(\B\theta',\B\sigma)$ under the type vector $\B\theta$ is
\begin{equation}\label{eq:welfare}
    \mathcal W(\B\theta',\B\sigma \mid \B\theta) := \sum_{a\in A} u(a,\B\theta',\B\sigma\mid \theta_a) + \sum_{b\in B} u(b,\B\theta',\B\sigma)
\end{equation}
\noindent
Since Resonance is strongly budget balanced and all participants are quasilinear, this is equal to the sum of valuations (minus costs) of all allocated transactions and nodes. We can therefore define the welfare of an allocation rule $\alpha$ as
\[\mc W(\alpha\mid \B\theta)=\sum_{t\in T(\alpha)} \theta_t - \sum_{n\in N(\alpha)} \theta_n(\alpha^{-1}(n)).\]
The welfare of a routing $R$ is the welfare of its allocation rule.

A welfare-maximizing strategy profile guarantees that the most economic value is generated by the resulting routing, but provides no information about how that value is distributed among different players. In particular, it is possible for a strategy profile to be welfare-maximizing and result in zero utility for all transactions and nodes (but positive utility for brokers). For our application (at Ritual), the primary efficiency objective is to maximize the economic surplus for transactions and nodes, while paying brokers the minimal amount necessary to offset their expenses. 

\paragraph{Surplus.} The \emph{surplus} $\mathcal S$ of an action profile $(\B\theta',\B\sigma)$ is the sum of utilities of transactions and nodes in the resulting routing:
\[\mathcal S(\B\theta',\B\sigma \mid \B\theta) := \sum_{a\in A} u(a,\B\theta',\B\sigma\mid \theta_a) . \]
Observe that surplus is equal to the first sum in \cref{eq:welfare}; the welfare of a routing is the sum of its surplus and its margin. Maximizing surplus guarantees \emph{both} that the most economic value is generated \emph{and} that this value is directed back to transactions and nodes.

The efficiency guarantees of Resonance hold \emph{in equilibrium}, which we define below as a refresher.

\begin{definition}[PNE]\label{def:pne}
An action profile $(\B\theta',\B\sigma) \in \B\theta\times\Sigma$ is a \emph{Pure Nash Equilibrium (PNE)} under type vector $\B\theta$ if
\begin{itemize}[wide=0.5em, leftmargin =*]
    \item for all agents $a\in A$, $\theta'_a$ is a best response to $(\B\theta'_{-a},\B\sigma)$, i.e.
    \[
        u(a,  \B\theta',\B\sigma\mid \theta_a) \ge u(a, \theta''_a,\B\theta'_{-a},\B\sigma \mid \theta_a) \enskip\enskip \forall \enskip\theta''_a \in \Theta_a;
    \]
    \item for all brokers $b\in B$, $\sigma_b$ is a best response to $(\B\theta',\B\sigma_{-b})$, i.e.
    \[
        u(b,  \B\theta',\B\sigma) \ge u(b, \B\theta',\sigma'_b,\B\sigma_{-b}) \enskip\enskip \forall \enskip\sigma'_b \in \Sigma_b.
    \]
\end{itemize}
\end{definition}
\noindent
In words, at a PNE, no one has an incentive to unilaterally deviate from their strategy. 
Not all games have PNEs, but Resonance has many. In particular, it follows as a corollary of \cref{thm:IC} that Resonance has PNEs.




We show in \cref{thm:welfare-1} and \cref{thm:welfare-2} that the PNEs of Resonance where transactions and nodes report their types truthfully are welfare-maximizing. Furthermore, while these PNEs achieve zero surplus if there is only one broker (\cref{thm:welfare-1}), they maximize surplus when there are there are at least two competing brokers (~\cref{thm:welfare-2}).

\begin{theorem}[Efficiency with one broker]\label{thm:welfare-1}
    Fix a type vector $\B\theta$ for transactions and nodes, and let $(\B\theta,\sigma_b)$ be a PNE of Resonance when $B=\set{b}$. Then
    \begin{itemize}[wide=0.5em, leftmargin =*]
        \item $(\B\theta,\sigma_b)$ is welfare-maximizing among all valid allocations: \[\mc W(\B\theta,\sigma_b\mid \B\theta)=\max_{\alpha\in V} \mc W(\alpha \mid \B\theta).\]
        \item $(\B\theta,\sigma_b)$ achieves zero surplus: $\mc S(\B\theta,\sigma_b\mid \B\theta)=0$.
    \end{itemize}
\end{theorem}
\begin{proof}
    We start by showing the following claims.

    \begin{claim}\label{claim:marginleqwelfare}
       If nodes and transactions report their types honestly, the margin of the output routing is upper bounded by the welfare of that routing. 
    \end{claim}
    \begin{proof}
        Let $R$ be the output routing under the action profile $(\B\theta,\B\sigma_B)$. We have $\mc W(R\mid \B\theta) =\mc S(R\mid \B\theta)+ \gamma(R)$, so it suffices to show that $\mc S(R\mid \B\theta)\geq 0$.
        
        If there was no winning broker, the mechanism outputs the empty allocation with zero payments, so $\mc S(R\mid \B\theta)=\mc W(R\mid\B\theta)=0$. Otherwise, there is a winning broker $b$ who submitted $R$. Since $R$ was accepted, it must have produced non-negative utility for all transactions and nodes under $\B\theta$. It follows that $R$ has non-negative surplus.
    \end{proof}

    \begin{claim}\label{claim:max-extract-routing}
        For any valid allocation rule $\alpha$, there exists a budget balanced routing with margin $\mc W(\alpha\mid\B\theta).$
    \end{claim}
    \begin{proof}
        This can be done by setting the allocation rule to $\alpha$, the payment of every transaction $t$ allocated by $\alpha$ to $v_t$, and the payment to every node $n$ to $c_n(\alpha^{-1}(n))$.
    \end{proof}

    We will now prove the first bullet. Let $R$ be the routing output by the mechanism under $(\B\theta,\sigma_b)$ (which was either submitted by $b$, or is the empty routing). The utility of $b$ under $(\B\theta,\sigma_b)$ is zero if $b$ does not win, and $\gamma(R)$ otherwise. It follows that 
    \begin{equation}\label{eq:utilleqwelfare}
        u(b,(\B\theta,\sigma_b)\mid\B\theta) \leq  \gamma(R) \leq \mc W(R\mid\B\theta),
    \end{equation}
    where the second inequality follows from \cref{claim:marginleqwelfare}.
    
    Assume for the sake of contradiction that $R$ is not welfare-maximizing, and let $\alpha$ be a welfare maximizing valid allocation rule with strictly higher welfare. By \cref{claim:max-extract-routing}, there is a budget-balanced routing $R'$ whose margin is $\mc W (\alpha\mid\B\theta)$. If $b$ proposes $R'$ instead of $R$, $b$ will win and collect a payment of $\mc W(\alpha\mid\B\theta)$. That is, 
    \[u(b,\B\theta,R'\mid \B\theta)=\mc W(\alpha\mid\B\theta) > \mc W(R\mid\B\theta) \geq u(b,\B\theta,\sigma_b\mid\B\theta),\]
    where the first inequality follows by definition of $\alpha$, and the second inequality follows from \cref{eq:utilleqwelfare}. This contradicts the best-response condition for $b$, concluding the proof of the first bullet.
    
    Next, we show the second bullet. Again, let $b$ be the only broker. Let $R$ be the routing output by the mechanism under $(\B\theta,\sigma_b)$ with allocation rule $\alpha$. We already know from the previous bullet that $\alpha$ is welfare-maximizing. Suppose for the sake of contradiction that $\mc S(R\mid \B\theta)\ne 0$. This means $R$ cannot be the empty routing, and therefore must have been proposed by $b$. Therefore, $u(b,\B\theta,\B\sigma\mid \B\theta)=\gamma(R)=\mc W(\alpha\mid \B\theta)-\mc S(R\mid \B\theta)$. However, consider the budget-balanced routing $R'$ with allocation rule $\alpha$ and margin $\mc W(\alpha\mid\B\theta)$ (which exists by \cref{claim:max-extract-routing}). If $b$ proposes $R'$ instead of $R$, it would win and get a payment of $\mc W(\alpha\mid\B\theta)$. That is,
    \[u(b,\B\theta,R'\mid\B\theta)=\mc W(\alpha\mid\B\theta) > W(\alpha\mid \B\theta)-\mc S(R\mid \B\theta) = u(b,\B\theta,\sigma_b\mid\theta),\]
    where the inequality follows from the assumption that $R$ has non-zero surplus. Bullet two follows.
\end{proof}

\begin{theorem}[Efficiency with multiple brokers]\label{thm:welfare-2}
    Fix a type vector $\B\theta$ for transactions and nodes, and let $\B\sigma = (\B\theta,\B\sigma_B)$ be a PNE of Resonance with $|B|\geq 2$. Then
    \begin{itemize}[wide=0.5em, leftmargin =*]
        \item $(\B\theta,\B\sigma)$ is welfare-maximizing among all valid allocations:\\ $\mc W(\B\theta,\B\sigma\mid \B\theta)=\max_{\alpha\in V} \mc W(\alpha \mid \B\theta)$.
        \item $(\B\theta,\B\sigma)$ is surplus-maximizing among all weakly budget-balanced routings:\\ $\mc S(\B\theta,\B\sigma\mid \B\theta)=\max\set{\mc S((\alpha,\pi,\phi)\mid \B\theta) \mid \alpha\in V \text{ and $(\alpha,\pi,\phi)$ is weakly budget balanced} }.$
    \end{itemize}
\end{theorem}
\begin{proof}
    We start by showing the following claim.
    \begin{claim}\label{claim:zero-margin}
    When $|B|\geq 2$, the output routing under $(\B\theta,\B\sigma)$ has zero margin.
    \end{claim}
    \begin{proof}
        Let $R=(\alpha,\pi,\phi)$ be the output routing under $(\B\theta,\B\sigma)$, and call $b$ be the broker who submitted it. Suppose $R$'s margin is not zero. The margin of $R$ cannot be negative since that would result in negative utility for $b$, which is strictly dominated by submitting the empty routing. Therefore, $R$ must have some strictly positive margin $\epsilon$. This means there exists some transaction $t$ whose payment $\pi(t)$ under $R$ is strictly positive. Let $\delta$ be a real number in the interval $\big(0,\min\set{\epsilon,\pi(t)}\big)$.
    
        We will construct a new routing $R'$ that is identical to $R$ except that its payment for $t$ is $\pi(t)-\delta$. The routing $R'$ is budget-balanced, has positive margin ($\gamma(R')=\epsilon-\delta>0$), and its surplus exceeds $R$'s surplus ($\mc S(R'\mid\B\theta) = \mc S(R\mid \B\theta) + \delta$).
        Let $b'\not=b$ be another broker who is not selected under $(\B\theta,\B\sigma)$. By submitting $R'$ instead of its current proposal $\sigma_{b'}$, broker $b'$ could win and collect a payment of $\gamma(R')$ instead of zero. This contradicts the fact that $(\B\theta,\B\sigma)$ is a Nash Equilibrium.
    \end{proof}
    
    We will show the second bullet first. Let $R$ be the output routing under $(\B\theta,\B\sigma)$, and call $b$ be the broker who submitted it. We know $R$ is budget-balanced by \cref{lem:BB}.
    
    Assume for the sake of contradiction that there exists an alternative budged-balanced routing $R'=(\alpha',\pi',\phi')$ that achieves strictly higher surplus than $R$. We know from \cref{claim:zero-margin} that $R$ has zero margin, and $R'$ is budget-balanced so it has non-negative margin. Therefore, the margin of $R'$ is weakly greater than the margin of $R$. We consider two cases: 

    \begin{case*} 
    [$R'$ has strictly higher margin than $R$]
    This contradicts the fact that $(\B\theta,\B\sigma)$ is a Nash equilibrium, since broker $b$ could submit $R'$ instead of $R$ and still win, strictly increasing its utility due to the higher margin.
    \end{case*}
    \begin{case*}[$R$ and $R'$ have equal margin]
        The proof of this case is similar to the proof of \cref{claim:zero-margin}.
        We know $R'$ has strictly positive surplus (call it $\epsilon$), so it must assign at least one transaction $t$ with valuation $\theta_t> 0$ (to a non-empty set of nodes).  We will construct a new routing $R''$ that is identical to $R'$ except that its payment for $t$ is $\pi'(t)-\delta$ for some $\delta \in(0,\min\set{\epsilon,\pi'(t)})$. Routing
        $R''$ has strictly higher surplus and strictly higher margin than $R$, meaning any broker $b'\not= b$ is strictly better off by deviating from $(\B\theta,\B\sigma)$ to reporting $R''$. This contradicts the best response condition for $b'$ in the Nash Equilibrium $(\B\theta,\B\sigma)$.
    \end{case*}
    \noindent
    It follows that $R$ is a surplus-maximizing routing for the type vector $\B\theta$, completing the proof of the second bullet.

    It remains to show the first bullet. We just showed that the output routing $R$ under $(\B\theta,\B\sigma)$ is surplus-maximizing, and we know from \cref{claim:zero-margin} that it has zero margin. Since the welfare of a routing is the sum of its surplus and margin, it follows that $R$ is also welfare-maximizing, completing the proof.
\end{proof}

\subsection{Incentive Compatibility}

In general, incentive compatibility is a desired property since it simplifies the experience of participants by limiting the amount of strategization they need to do. While broker strategization is desirable to aide in discovering high-welfare routings, it is important that this does not come at the cost of complexity for transactions and nodes, who are often unsophisticated agents. The strongest version of incentive compatibility is \emph{dominant-strategy incentive compatibility (DSIC)}, which roughly states that it is optimal for each transaction and node to report their types truthfully, regardless of the types and actions of other participants.

\begin{definition}[DSIC]\label{def:dsic}
    Let $a\in A$. We say Resonance is \emph{dominant-strategy incentive compatible (DSIC) for $a$} if for all types $\theta_a\in\Theta_a$ and all strategies $\B\theta'_{-a}\in \B\Theta_{-a},\B\sigma\in\B\Sigma$ of other agents, $\theta_a$ is a best response to $(\B\theta'_{-a},\B\sigma)$ when $a$ has type $\theta_a$. 
\end{definition}

This strong notion of incentive compatibility is in general difficult to achieve (see, \textit{e.g.}, \cite{green-laffont-1979,MS}), especially when the strategy space of participants is sufficiently rich. In particular, Resonance is not DSIC for transactions and nodes when there are multiple brokers.
We will instead consider the following (weaker) version of incentive-compatibility. It states that truthful bidding is optimal for transactions and nodes assuming brokers play a fixed strategy. This means transactions and nodes do not need to reason about the types and strategies of \emph{other transactions and nodes} as long as brokers play certain equilibrium strategies.


\begin{definition}[DSIC barring $B$]\label{def:dsic-barring-s}
    We say Resonance mechanism is \emph{DSIC barring $B$} if for all type vectors $\B\theta\in \B\Theta$, there exists a broker action profile $\B\sigma\in\B\Sigma$ such that
    \begin{itemize}[wide=0.5em, leftmargin =*]
        \item for all agents $a\in A$, all types $\theta_a\in \Theta_a$, and all actions $\B\theta'_{-a}\in \B\Theta_{-a}$ of other agents in $A$, $\theta_a$ is a best response to $(\B\theta'_{-a},\B\sigma)$ when $a$ has type $\theta_a$; and
        \item $(\B \theta, \B \sigma)$ is a PNE.
    \end{itemize}
\end{definition}
\noindent
Here is one way of interpreting \cref{def:dsic-barring-s}.
Fixing the actions $\B\sigma$ of brokers in Resonance induces a subgame between transactions and nodes. The first bullet in \cref{def:dsic-barring-s} requires that, for \emph{some} $\B\sigma$, the induced subgame is DSIC for all transactions and nodes. The second bullet requires that $\B\sigma$ forms an equilibrium with honest bids by transactions and nodes.
In other words, \cref{def:dsic-barring-s} requires that in a subset of the PNEs of Resonance, transactions and nodes do not need to strategize with respect to the actions of other transactions and nodes. We will characterize the particular subset of PNEs for which this definition holds in \cref{thm:IC}, and then argue why these are a natural subset at the end of this section.


The following theorem outlines the incentive properties of Resonance. When there is only one broker, we show that Resonance is DSIC for transactions and nodes. Intuitively, having one broker essentially creates a posted price experience for transactions and nodes, each of whom must choose between the broker's proposal and the empty routing. 
When there are multiple brokers, transactions and nodes no longer face a take-it-or-leave-it offer if brokers submit routings \emph{with different allocations}. Since the mechanism's process of selecting an allocation relies on the submissions of transactions and nodes, there can be beneficial strategization by transactions and nodes who prefer one broker's allocation over another's. If all brokers were guaranteed to submit the same allocation, we would recover a (conditional) posted-price experience for transactions and nodes. This is captured formally in the property of DSIC barring $B$, which Resonance satisfies even in the case of multiple brokers.

\begin{theorem}\label{thm:IC}
Resonance is DSIC for transactions and nodes when $|B|=1$, and DSIC barring $B$ when $|B|\geq 2$.
\end{theorem}
\begin{proof}
    First, we consider the case where there is only one broker. 
    %
    %
    %
    Let $a$ be a transaction or node with type $\theta_a$, and fix the strategies of other agents $(\B\theta'_{-a},\B\sigma)$. Denote by $R_b=\sigma_b$ be the routing submitted by the broker. No matter what $a$ plays, the possible outcomes of the mechanism are either $R_b$ or the empty routing with zero payments. In the former case, and $a$ achieves utility $u(a,R_b\mid \theta_a)$, and in the latter case it achieves utility 0.

    Let $R$ be the output of the mechanism when $a$ plays honestly.
    \begin{case*}[$R=R_b$]
        If $R=R_b$, we must have $u(a,R_b\mid \theta_a)\geq 0$, since Resonance checks that its output does not result in non-negative utility according to the reported types of all participants. This means $a$'s utility is the higher of its two possible values, so $\theta_a$ is a best response to  $(\B\theta'_{-a},\B\sigma)$ and we are done.
    \end{case*}
    \begin{case*}[$R\ne R_b$]
        If $R\ne R_b$, then $R$ must be the empty routing, because either $R_b$ was not budget balanced, or it resulted in negative utility for some $a'\in A$.
        \begin{itemize}[nosep,leftmargin=*]
            \item If either $R_b$ is not budget-balanced or $a\ne a'$,  the mechanism outputs the empty routing no matter what $a$ does, so $\theta_a$ is a best response.
            \item If $a=a'$, then $u(a,R_b\mid \theta_a) < 0 =u(a,R\mid \theta_a)$, so $\theta_a$ is a best response and we are done.
        \end{itemize}
    \end{case*}
    
    We will next consider the case of multiple brokers. We fix a type vector $\B\theta\in\Theta$ and construct the broker action profile $\B\sigma$ as follows. Recall that brokers know the valuations and costs of all transactions and nodes. For all brokers $b\in B$, we set $\sigma_b$ to a surplus-maximizing routing (breaking ties arbitrarily but consistently among all brokers).\footnote{\label{footnote:general-pne}In fact, it is not necessary that $\sigma_b$ is the same routing for all $b$. As long as all routings $\sigma_b$ have the same \emph{allocation} (with arbitrary and potentially different payment rules for transactions and nodes), the first bullet in the definition of ``DSIC barring $B$'' is satisfied. If in addition this allocation is welfare-maximizing and at least two broker proposals have zero margin, the second bullet in the definition of ``DSIC barring $B$'' is also satisfied. The proof as written works under the assumptions of this footnote.} Note that since the tie-breaking is consistent across brokers, all routings in $\sigma_b$ have the same allocation, which we will denote by $\alpha$. 
    
    First, let us confirm that the first bullet in \cref{def:dsic-barring-s} is satisfied.
    Let $a\in A$ with type $\theta_a$, and fix a strategy profile $\B\theta'_{-a}$ for the remaining transactions and nodes. We must show that $\theta_a$ is a best response to $(\B\theta'_{-a},\B\sigma)$.

    Denote by $R$ the routing output by Resonance under the strategy profile $\big(\theta_a, {\B\theta}'_{-a},\B\sigma)$.
    We claim that no matter what $a$ plays, the mechanism outputs either $R$ or the empty routing. This is because all routings in $\B\sigma$ have the same allocation rule, and therefore the same surplus under any type vector, and the mechanism's tie-breaking rule among surplus-maximizing allocations is independent of the submissions of participants. We have therefore reduced the problem to the case of $|B|=1$, where we have already shown the mechanism is DSIC for $a$.

    Finally, we confirm that the second bullet in \cref{def:dsic-barring-s} is satisfied. That is, we must check that $(\B\theta,\B\sigma)$ is a PNE.
    The best-response conditions for transactions and nodes follow from the first bullet. To show the best-response condition for brokers, we will show that a broker achieves non-positive utility in any unilateral deviation: For all $b$, and all strategies $R\in \Sigma$, we have $u(b, \B\theta,R,{\B\sigma}_{-b})\leq 0$.


     Suppose for the sake of contradiction that there is a proposal $R$ that results in positive utility for $b$. Then the mechanism must output $R$ under $(\B\theta, R,{\B \sigma}_{-b})$, and $R$ must have margin $\gamma(R)>0$. Let $b'\ne b$ be another broker (which exists since $|B|\geq 2$). By construction, $\sigma_{b'}$ is surplus-maximizing, so it is welfare-maximizing with zero margin, meaning 
    \begin{align*}
        \mc S (\sigma_{b'}\mid \B\theta) = \mc W (\sigma_{b'}\mid \B\theta) \geq \mc W (R\mid \B\theta) = \mc S (R\mid \B\theta) + \gamma(R)
    \end{align*}
    The right-hand-side is strictly smaller than $\sigma_{b'}$ when $\gamma(R)>0$, implying the mechanism would not output $R$, which is a contradiction.

    
        
\end{proof}

We conclude the discussion on incentives with a few remarks.
A possible criticism of \cref{thm:IC} is that the incentive guarantees in the presence of multiple brokers are relatively weak. In particular, DSIC barring $B$ holds as long as there \emph{exists} a broker action profile that induces a DSIC subgame for transactions on nodes. Why should we expect that this particular action profile will be the one brokers actually play in practice?

In response, we first point out that the proof of \cref{thm:IC} holds for a \emph{set} of broker action profiles. As explained in \cref{footnote:general-pne}, the conditions of \cref{def:dsic-barring-s} are satisfied by any broker action profile in which all broker proposals have the same welfare-maximizing allocation (albeit at different prices) and at least two brokers have zero margin. We will next argue that it is reasonable to expect brokers to play one of these action profiles in practice, since they are the action profiles that would arise as a result brokers competing by undercutting each other's margins.

This can be formalized in the following sense. Consider the broker subgame induced by Resonance after fixing the strategies of transactions and nodes to honestly reporting their types. Consider best-response dynamics in this subgame, which start with a broker action profile, and repeatedly select a broker to best-respond to the current actions of other brokers. Assuming the welfare-maximizing allocation is unique,\footnote{We expect that the welfare-maximizing allocation will be unique in practice, since the set of instances that yield multiple optimal allocations is a small fraction of all instances (\textit{e.g.} has measure zero when valuations/costs are unbounded real numbers).} and the broker margins in the starting action profile are not too small,\footnote{When capturing iterative competition among brokers, it is reasonable to assume the initial strategy of each broker is to ask for a large margin.} best-response dynamics converge precisely to the set of actions profiles where every broker plays the welfare-maximizing allocation, with two brokers having close to zero margin. This is because, whenever the current winning broker has a positive margin, the best-response by a losing broker is to submit a welfare-maximizing allocation with slightly lower margin than the current winner's. As long as the initial margins are large enough to allow for enough iterations for all brokers to be selected at least once before anyone reaches zero margin, the terminal state will include only routings with the welfare-optimal allocation. Furthermore, unless the margin of at least two brokers is within one increment away from zero, there exists a broker who is not best-responding so the dynamics have not yet converged.

Note that the only equilibrium action profiles of brokers that may violate \cref{def:dsic-barring-s} are ones where a) there is a winning broker submitting a welfare-maximizing allocation with zero margin, and b) there are some losing brokers who play different allocations. The best-response dynamics argument above aims to show that these equilibria are in some sense ``degenerate.''

\subsection{Tractability}
Resonance is simple to describe, and computationally light-weight. The computational steps involve calculating the surplus of each proposal, finding the proposal with maximum surplus, and checking that it results in non-negative utility for all agents. All of these steps can be done in time proportional to $|A|\cdot |B|$ (assuming agent types have constant description size).
\section{Practical Considerations and Future Work}\label{sec:conclusion}

As shown in Section \ref{sec:broker}, Resonance satisfies several desirable properties, including budget-balance, individual rationality, efficiency, incentive compatibility, and tractability. However, there are other concerns that are less well-studied in the academic literature but are relevant in making sure this mechanism can work well in practice. We informally outline some of these broad topics in this section; we will rigorously address them in future work. 

\paragraph{Griefing Vectors.} The mechanism as stated lends itself to a simple analysis. However, the mechanism exactly as stated is vulnerable to ``griefing'', in the sense that the cost of reducing the efficiency of the network is small. We outline two such attacks. 
\begin{enumerate}
    \item \textit{Rejection due to a Single Agent.} Under the specification as given in Section \ref{sec:broker-spec}, the mechanism outputs the empty routing if $R^*$ results in negative utility for \textit{any} transaction or node. Thus, a transaction or node with an improper bid can lead the mechanism to generate zero surplus. In a practical implementation, broker proposals that otherwise generate high surplus, but fail to achieve non-negative surplus for a particular transaction or node must be amended to yield partial surplus. It is indeed possible to lightly modify the mechanism to accomplish this while preserving the properties outlined in Section \ref{sec:broker}.
    \item \textit{Private Order Flow.} In the presence of private order flow, a malicious broker can cheaply cause the mechanism to output a low surplus matching by allocating a fake transaction that only it privately sees for a low cost. By setting the valuation of that transaction to be very high, allocating the fake transaction can yield higher elicited surplus than any other proposal, thus causing only the fake transaction to be allocated. Currently, there is no private order flow in the mechanism's implementation. Certain modifications to the mechanism that accommodate specialized brokers (as outlined next in this section, and will be addressed in  future work) are unaffected by this attack. 
\end{enumerate}

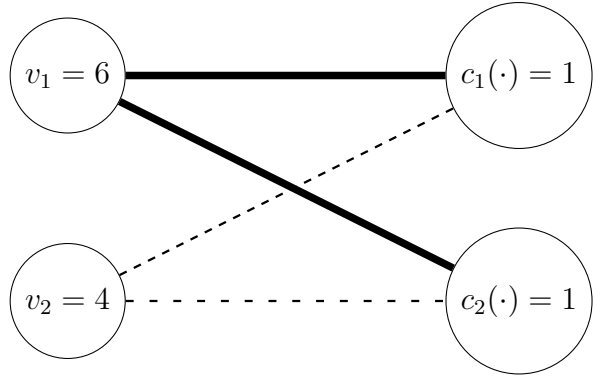
\begin{figure}[t]\label{fig:6411}
    \centering
    \begin{tikzpicture}
\node[shape=circle,draw=black] (1) at (-3,0) {$v_2 = 4$};
\node[shape=circle,draw=black] (2) at (-3,3) {$v_1 = 6$};
\node[shape=circle,draw=black] (3) at (3,0) {$c_2(\cdot) = 1$};
\node[shape=circle,draw=black] (4) at (3,3) {$c_1(\cdot) = 1$};

\path[line width=1mm, -] (2) edge (3);
\path[line width=1mm, -] (2) edge (4);

\path[line width=0.3mm, dashed] (1) edge (4);
\path[line width=0.3mm, loosely dashed] (1) edge (3);
\end{tikzpicture}
    \caption{\textbf{Nonexistence of Collusion-Resistant Routings.} Consider the above example where there are two transactions and two nodes. The first transaction (shown on the top left) requires both nodes to execute, and the second transaction (shown on the bottom left) only requires one node, and can be executed with either node. Nodes however only have capacity to execute one of the two transactions. There are therefore three valid (non-empty) allocations: one in which the first transaction is allocated to both nodes, and two more corresponding to the second transaction being allocated to either the first or second node. The surplus maximizing allocation maps the first transaction to both nodes, yielding a surplus of $6-1-1 = 4$. However, if the first transaction pays one of the nodes less than $4 = v_2$, then that node and the second transaction would prefer to collude, since the second transaction is willing to pay up to $4$. Because $4 + 4 > 6 = v_1$, the first transaction cannot pay both nodes at least four, whence no routing is collusion resistant.}
    \label{fig:collusion}
\end{figure}

\paragraph{Broker Specialization.} In the mechanism specification as outlined in Section \ref{sec:broker-spec}, brokers propose an allocation over \textit{all} transactions and nodes. If brokers specialize and are only informed about the valuations/costs of transactions and nodes of a certain type, the mechanism will not be able to make optimal use of the brokers' information. Concretely, if transactions can be partitioned into two disjoint categories, where some brokers are informed about transaction valuations of the first category and other brokers are informed of valuations of the second category, the mechanism will not be able to make use of both sets of brokers' information concurrently. If it is known to the mechanism designer up-front what such categories may exist (and validity cleanly decomposes across these categories as well), then parallel versions of Resonance of Section \ref{sec:broker} can be run over each category to maximize surplus. However, in the more general case where brokers may have information over an arbitrary subset of the transactions and nodes, it becomes necessary to allow brokers to submit partial allocations, and devise a scheme by which those partial allocations are stitched together to yield a full allocation. A modification to this mechanism accomplishes this, which we intend to outline in future work. Efficiency and incentive compatibility analysis is complicated in this regime, as it depends on the structure of overlapping specialization between competing brokers.

\paragraph{Broker Compensation.}
In our model, brokers are endowed with knowledge of the agents' type vector at no cost, and are assumed to be able to compute optimal routings given those types. In practice, both type estimation and welfare maximizing are complex and costly tasks, and brokers must be compensated accordingly to be incentivized to participate in the mechanism. Our mechanism can be extended to handle the case of non-zero brokers costs, where the amount of extracted margin in equilibrium will be reflect the cost of operating a broker. We defer the analysis of that case to future work. 

\paragraph{Collusion and Sybil Resistance.} Another concern that has not been discussed in the analysis of Section \ref{sec:broker} is collusion and sybil resistance. There are multiple formulations of these notions relevant to the blockchain setting, ranging from the standard definition of group strategyproofness from the mechanism design literature, to notions of off-chain agreements that were first defined in \cite{tfm-1559, tfm-foundations} and further studied in \cite{tfm-collusion, gafni2024barriers}. Due to the complexity of our matching setting, it is impossible to obtain the strictest form of these guarantees (see Figure \ref{fig:collusion} for an example). However, a relaxed guarantee may indeed be possible to accomplish.

\bibliographystyle{plain} 
\bibliography{bib}

\begin{thebibliography}{10}

\bibitem{zksync}
zksync docs.
\newblock \url{https://docs.zksync.io/build}.

\bibitem{gelato}
The all-in-one ethereum rollup as a service platform.
\newblock \url{https://www.gelato.network/raas}, 2022.

\bibitem{gensyn-litepaper}
The hyperscale, cost-efficient compute protocol for the world’s deep learning models.
\newblock \url{https://docs.gensyn.ai/litepaper}, 2022.

\bibitem{raas}
Rollup as a service: Opportunities and challenges.
\newblock \url{https://ethresear.ch/t/rollup-as-a-service-opportunities-and-challenges/13051}, 2022.

\bibitem{multidim2}
Guillermo Angeris, Theo Diamandis, and Ciamac Moallemi.
\newblock Multidimensional blockchain fees are (essentially) optimal.
\newblock {\em arXiv preprint arXiv:2402.08661}, 2024.

\bibitem{scroll-costs}
Andy Arditi and Ye~Zhang.
\newblock The anatomy of proof generation.
\newblock \url{https://scroll.io/blog/proof-generation}, 2022.

\bibitem{chainlink}
Lorenz Breidenbach, Christian Cachin, Benedict Chan, Alex Coventry, Steve Ellis, Ari Juels, Farinaz Koushanfar, Andrew Miller, Brendan Magauran, Daniel Moroz, et~al.
\newblock Chainlink 2.0: Next steps in the evolution of decentralized oracle networks.
\newblock {\em Chainlink Labs}, 1:1--136, 2021.

\bibitem{multidim4}
Vitalik Buterin.
\newblock Multidimensional eip 1559.
\newblock \url{https://ethresear.ch/t/multidimensional-eip-1559/11651}, 2022.

\bibitem{multidim3}
Vitalik Buterin.
\newblock Multidimensional gas pricing.
\newblock \url{https://vitalik.eth.limo/general/2024/05/09/multidim.html}, 2024.

\bibitem{4844}
Vitalik Buterin, Dankrad Feist, Diederik Loerakker, George Kadianakis, Matt Garnett, Modi Taiwo, and Ansgar Dietrichs.
\newblock Eip-4844: Shard blob transactions.
\newblock \url{https://eips.ethereum.org/EIPS/eip-4844}, 2022.
\newblock [Online serial]. Available: https://eips.ethereum.org/EIPS/eip-4844.

\bibitem{subset-sum}
Alberto Caprara, Hans Kellerer, and Ulrich Pferschy.
\newblock The multiple subset sum problem.
\newblock {\em SIAM Journal on Optimization}, 11(2):308--319, 2000.

\bibitem{zama-whitepaper}
Ilaria Chillotti, Marc Joye, and Pascal Paillier.
\newblock Programmable bootstrapping enables efficient homomorphic inference of deep neural networks.
\newblock In {\em Cyber Security Cryptography and Machine Learning: 5th International Symposium, CSCML 2021, Be'er Sheva, Israel, July 8--9, 2021, Proceedings 5}, pages 1--19. Springer, 2021.

\bibitem{tfm-collusion}
Hao Chung, Tim Roughgarden, and Elaine Shi.
\newblock Collusion-resilience in transaction fee mechanism design.
\newblock {\em arXiv preprint arXiv:2402.09321}, 2024.

\bibitem{tfm-foundations}
Hao Chung and Elaine Shi.
\newblock Foundations of transaction fee mechanism design.
\newblock In {\em Proceedings of the 2023 Annual ACM-SIAM Symposium on Discrete Algorithms (SODA)}, pages 3856--3899. SIAM, 2023.

\bibitem{multidim1}
Theo Diamandis, Alex Evans, Tarun Chitra, and Guillermo Angeris.
\newblock Designing multidimensional blockchain fee markets.
\newblock In {\em 5th Conference on Advances in Financial Technologies (AFT 2023)}. Schloss Dagstuhl-Leibniz-Zentrum f{\"u}r Informatik, 2023.

\bibitem{gafni2024barriers}
Yotam Gafni and Aviv Yaish.
\newblock Barriers to collusion-resistant transaction fee mechanisms.
\newblock {\em arXiv preprint arXiv:2402.08564}, 2024.

\bibitem{green-laffont-1979}
Jerry Green and Jean-Jacques Laffont.
\newblock On coalition incentive compatibility.
\newblock {\em The Review of Economic Studies}, 46(2):243--254, 1979.

\bibitem{matic}
Jaynti Kanani, Sandeep Nailwal, and Anurag Arjun.
\newblock Matic whitepaper.
\newblock {\em Polygon, Bengaluru, India, Tech. Rep., Sep}, 2021.

\bibitem{cosmos}
Jae Kwon and Ethan Buchman.
\newblock Cosmos whitepaper.
\newblock {\em A Netw. Distrib. Ledgers}, 27:1--32, 2019.

\bibitem{lagrange-paper}
Arthur Lazzaretti, Charalampos Papamanthou, and Ismael Hishon-Rezaizadeh.
\newblock Robust double auctions for resource allocation.
\newblock {\em Cryptology ePrint Archive}, 2024.

\bibitem{chaos}
S~Leonardos, D~Reijsbergen, B~Monnot, and G~Piliouras.
\newblock Optimality despite chaos in fee markets. financial cryptography and data secutiry’23, 2022.

\bibitem{myerson1981optimal}
Roger~B Myerson.
\newblock Optimal auction design.
\newblock {\em Mathematics of operations research}, 6(1):58--73, 1981.

\bibitem{MS}
Roger~B Myerson and Mark~A Satterthwaite.
\newblock Efficient mechanisms for bilateral trading.
\newblock {\em Journal of economic theory}, 29(2):265--281, 1983.

\bibitem{bittensor}
Yuma Rao, Jacob Steeves, Ala Shaabana, Daniel Attevelt, and Matthew McAteer.
\newblock Bittensor: A peer-to-peer intelligence market.
\newblock {\em arXiv preprint arXiv:2003.03917}, 2020.

\bibitem{tfm-1559}
Tim Roughgarden.
\newblock Transaction fee mechanism design.
\newblock {\em ACM SIGecom Exchanges}, 19(1):52--55, 2021.

\bibitem{modulus}
Salvatore Vivona and Luca Vivona.
\newblock Modulus: An open modular design for interoperable and reusable machine learning.

\bibitem{MD-for-prover-markets}
Wenhao Wang, Lulu Zhou, Aviv Yaish, Fan Zhang, Ben Fisch, and Benjamin Livshits.
\newblock Mechanism design for zk-rollup prover markets.
\newblock {\em arXiv preprint arXiv:2404.06495}, 2024.

\bibitem{polkadot}
Gavin Wood.
\newblock Polkadot: Vision for a heterogeneous multi-chain framework.
\newblock {\em White paper}, 21(2327):4662, 2016.

\bibitem{fhenix-whitepaper}
Guy Zyskind, Yonatan Erez, Tom Langer, Itzik Grossman, and Lior Bondarevsky.
\newblock Fhe-rollups: Scaling confidential smart contracts on ethereum and beyond.

\end{thebibliography}

\appendix
\section{On Multi-Dimensional Fee Mechanisms}\label{sec:mdfm}
This section provides a formal discussion on whether the multi-dimensional fee market design can be extended to the setting with both demand and supply heterogeneity. Mechanisms are responsible for doing two things: First, they determine the allocation, \textit{i.e.} which transactions are executed, and by which nodes (in the case that multiple nodes exist). Second, they determine the prices for execution that are paid by users, and the payments that are paid to node operators. Current literature on multi-dimensional fee markets focuses primarily on price-setting. However, the problem of finding an optimal allocation is complex in its own right. Even in a model with no supply-side cost heterogeneity and only one execution node (as in Ethereum, where this node corresponds to the chain itself), finding the optimal allocation requires solving a multi-dimensional knapsack problem \cite{multidim3}. No efficient algorithms are known to solve this problem, and unlike the one-dimensional case where a greedy approach yields a good approximation guarantee, the same is not true in the multi-dimensional setting. 

Despite these computational concerns, multi-dimensional fee markets generalize existing transaction fee mechanisms in a very simple way. Furthermore, in ``Multidimensional Blockchain Fees are (Essentially) Optimal'' \cite{multidim2} and ``Designing Multidimensional Blockchain Fee Markets'' \cite{multidim1}, it is shown that prices can be updated in a way that yields good welfare guarantees. However, this result assumes the existence of an oracle that can find the optimal allocation given a set of prices for each dimension. It is unclear how one would design an oracle (or whether such an oracle could theoretically exist); in major blockchains such as Ethereum, if multiple valid allocations exist among transactions that clear the base fee, a priority fee system is used to select the highest value transactions. However, if priority fees are used, the mechanism ceases to be \textit{incentive-compatible}, or strategically simple to engage with from the perspective of users. Is such complex user strategization or collusion necessary in order to obtain good outcomes?

To understand the limits of what economic efficiency is possible to achieve without making such assumptions, we formally separate the pricing and allocation steps of the multi-dimensional fee mechanism in Section \ref{sec:mdfm}. In the pricing step, unit base fees are set for each dimension, and all transactions that are unwilling to pay the base fee are screened out. In the allocation step, the mechanism must determine how the remaining transactions must get allocated. We consider a hierarchy of allocations:
\begin{enumerate}
    \item \textit{Inclusion-Maximal Allocations.} In an inclusion-maximal allocation, the mechanism allocates a maximal set of transactions to nodes. That is, any allocation that allocates a strict superset of the transactions that are allocated under an inclusion-maximal allocation must be invalid (\textit{e.g.} by exceeding resource constraints). However, besides maximality, these allocations need not satisfy any additional properties, such as base-fee revenue maximization. Inclusion-maximal allocations are natural to consider in settings where computation of a more advanced allocation is intractable, as would be the case if the number of dimensions in the market is high.
    \item \textit{Fee-Maximal Allocations.} In a fee-maximal allocation, the mechanism allocates transactions that maximize the total revenue from base fees. Computing a multi-dimensional fee-maximal allocation involves solving the \textit{multiple subset-sum} problem, which is NP-Hard, and further hard to approximate \cite{subset-sum}. Thus, these allocations may not be tractable to compute in more heterogeneous settings. However, base-fee revenue is a natural benchmark to consider if priority fees are not used; as the base fee is a lower bound for the transaction's valuation, a fee-maximal allocation maximizes a lower bound for transaction surplus.
    \item \textit{Oracle-Maximal Allocations.} An oracle-maximal allocation maximizes surplus given full knowledge of transaction valuations and node costs. It is unclear exactly how the mechanism would obtain such knowledge (though priority fees perhaps provide a coarse approximation), hence these allocations serve more as a theoretical benchmark. 
\end{enumerate}

We perform a worst-case analysis on the efficiency of each of these allocations assuming that a best-case pricing is chosen. That is, for each of these allocation classes, we define benchmarks that characterize the efficiency a multi-dimensional fee mechanism would achieve if pricing was done optimally (\textit{e.g.} if a perfect pricing update rule was used), but a worst-case allocation among those from the allocation class was chosen. As shown in \cref{thm:hierarchy}, under best-case pricing, the worst-case performance of inclusion-maximal allocations is weakly dominated by that of fee-maximal allocations, which is further weakly dominated by oracle-maximal allocations. 

We then focus more closely on the Ethereum setting, in which there is one node that can execute computation (the chain itself), and there are resource constraints for execution along each of $d$ dimensions. We first show in Theorem \ref{thm:of} that even under best-case pricing, inclusion-maximal allocations can obtain at most $\frac2d$ of optimal surplus in the worst case. Thus, without computing a more advanced allocation, efficiency can be arbitrarily bad as heterogeneity increases, even given best-case pricing. We show in Theorem \ref{thm:fee} that fee-maximal allocations also suffer in efficiency; for any $d \ge 1$, fee-maximal allocations can obtain at most $\frac12$ of optimal surplus in the worst case even given best-case pricing. 

Next, we consider the heterogeneous setting, where multiple nodes (\textit{e.g.} off-chain entities, or even nodes within the network) can execute computation, but incur heterogeneous unit costs for doing so. We show general impossibility of obtaining good surplus with a multi-dimensional fee mechanism in this regime. Formally, we show in Theorem \ref{thm:wo} that unless a unique dimension exists for each transaction (so that no two transactions have nonzero values across any shared dimensions), the surplus obtained under an oracle-maximal allocation may be arbitrarily poor relative to optimal even given best-case pricing. That is, even with best-case pricing and best-case allocation, the simple act of constraining the mechanism to price all transactions with a shared vector of unit base fees can cause surplus to be poor relative to optimal when there are multiple heterogeneous compute nodes. Thus, a fundamentally new type of mechanism is required to obtain good surplus in this setting.

\subsection{Analysis}
Multi-dimensional fee mechanisms clear multi-dimensional resource markets, where there is a set of distinct resources that each transaction may use, and each transaction can be associated with a vector that denotes the consumption of each resource required to execute the transaction. A multi-dimensional fee mechanism clears this resource market by assigning a vector of prices (one for each resource) that determine how much users must pay per unit usage of each resource. The mechanism need not allocate all transactions that clear this base fee, but it will reject any transactions that do not clear the base fee.  

\paragraph{Multi-Dimensional Resource Markets.} Formally, in a $d$-dimensional \textit{resource market}, we let each transaction $t \in T$ have an associated $d$-dimensional vector $\B g(t) \in \Rg^d$ with non-negative entries denoting the resource usage of the transaction along each of the $d$ dimensions. Transaction and node types are given by a vector $\B \theta$, which is omitted for clarity in what follows. There is a pre-determined set $V$ of valid allocations. We define by
\begin{equation}\label{eq:opt}
     \mathrm{OPT} := \max_{R\in\mc{R}} \set{\mc S(R)}
\end{equation}
the maximal surplus that can be attained by a valid routing for resource market.

\paragraph{Multi-Dimensional Fee Mechanisms.} A \textit{multi-dimensional fee mechanism} clears a $d$-dimensional resource market in two steps. First, in the \textit{pricing step}, the mechanism chooses a vector of unit base fees $\B{p} \in \Rg^d$. This vector denotes the unit costs/payments for each resource (i.e. for a transaction $t$, the total base fee is given by $\B g(t)^\top \B p$). Next, in the \textit{allocation step}, the mechanism chooses some allocation $\alpha$ that rejects all transactions that are unwilling to pay the base fee. That is, the mechanism chooses some allocation $\alpha \in V_{\B p}$, where
\begin{equation}\label{eq:vp}
    V_{\B p} := \set{\alpha \in V \mid \B g(t)^\top \B p > v_t \implies \alpha(t) = \emptyset \text{ for all } t \in T}
\end{equation}

\paragraph{Payment Rules.} In a multi-dimensional fee mechanism (without priority fees), after the allocation $\alpha \in V_{\B p}$ is chosen, any included transaction pays their base fee, and nodes receive the sum of the fees paid by transaction that were executed by them. That is, under a vector of unit base fees $\B p$ and the allocation $\alpha$, the transaction payment rule is given by  
\begin{equation}\label{eq:pip}
    \pi_\alpha^{\B p}(t) := \B g(t)^\top \B p \cdot \mathbbm{1}\bra{\alpha(t) \ne \emptyset}
\end{equation}
and the node payment rule is given by
\begin{equation}\label{eq:phip}
    \phi_\alpha^{\B p}(n) :=  \sum_{t \in \alpha^{-1}(n)} \B g(t)^\top\B p.
\end{equation}

Which allocation $\alpha \in V_{\B p}$ should be chosen? In what follows, we define three benchmark classes of allocations, and give a worst-case analysis for the welfare of each relative to $\mathrm{OPT}$ from \cref{eq:opt}.

\paragraph{Inclusion-Maximal Allocations.} The first benchmark class of allocations we define are \textit{inclusion-maximal} allocations. These are allocations that would arise if we were to arbitrarily allocate transactions to nodes until there are no more resources available. Formally, we define the set of all inclusion-maximal allocations $I_{\B p} \subseteq V_{\B p}$ to be
\begin{equation}
    I_{\B p} := \set{\alpha \in V_{\B p} \mid \nexists \beta \in V_{\B p} \text{ such that } T(\beta) \supset T(\alpha)}
\end{equation}
That is, an inclusion-maximal allocation $\alpha \in I_{\B p}$ satisfies the property that no other allocation $\beta \in V_{\B p}$ allocates a strict superset of transactions in $T$ that are allocated under $\alpha$. The maximum surplus that can be attained under a best-possible vector of unit base fees $\B p$, but a worst-case inclusion-maximal allocation in $I_{\B p}$ is given by
\begin{equation}\label{eq:optinc}
    \mathrm{INC} := \max_{\B p \in \Rg^d}\set{\min_{ \alpha \in I_{\B p}}\set{\mc S(\alpha, \pi_\alpha^{\B p}, \phi_{\alpha}^{\B p}) } }
\end{equation}
This quantity denotes the worst-case surplus we would achieve if we were to arbitrarily allocate transactions to nodes until no more resources are available, even if we had the best possible unit base fees $\B p$. 

\paragraph{Fee-Maximal Allocations.} The next benchmark class of allocations that we define are \textit{fee-maximal} allocations. These are allocations that would arise if we were to allocate transactions to nodes such that the total fees paid by users is maximized. Finding a fee-maximal allocation is computationally hard: no known efficient algorithms exist to do so. Furthermore, when $d \ge 2$, no known efficient algorithms exist to approximately compute these allocations. Formally, we define the set of all fee-maximal allocations $F_{\B p} \subseteq V_{\B p}$
\begin{equation}
    F_{\B p} := \arg\max_{ \alpha \in V_{\B p}}\set{\sum_{t \in T(\alpha)} \B g(t)^\top \B p}
\end{equation}
That is, a fee-maximal allocation maximizes the total base fees generated by included transactions $T(\alpha)$. The maximum surplus that can be attained under a best-possible vector of unit base fees $\B p$, but a worst-case fee-maximal allocation in $F_{\B p}$ is given by
\begin{equation}\label{eq:optrev}
    \mathrm{FEE} := \max_{\B p \in \Rg^d}\set{\min_{\alpha \in F_{\B p}}\set{\mc S(\alpha, \pi_\alpha^{\B p}, \phi_{\alpha}^{\B p})}}
\end{equation}

\paragraph{Oracle-Maximal Allocations.} The final benchmark class of allocations that we define are \textit{oracle-maximal} allocations. These correspond to the best possible allocations we can choose given oracle-level knowledge of agent valuations and costs. In practice, it is impossible to obtain this knowledge, though with priority fees we may obtain some approximation at the cost of truthfulness. Formally, we define the set of all oracle-maximal allocations $O_{\B p} \subseteq V_{\B p}$ as 
\begin{equation}
    O_{\B p} := \arg\max_{\alpha \in V_{\B p}}\set{\mc S(\alpha, \pi_{\alpha}^{\B p}, \phi_{\alpha}^{\B p}) }
\end{equation}
The maximum surplus that can be attained under a best-possible vector of unit base fees $\B p$, and any allocation in $O_{\B p}$ is given by
\begin{equation}\label{eq:optora}
    \mathrm{ORA} := \max_{\B p \in \Rg^d}\set{\max_{\alpha \in V_{\B p}} \set{\mc S(\alpha, \pi_\alpha^{\B p}, \phi_{\alpha}^{\B p})}}
\end{equation}

\begin{proposition} \label{thm:hierarchy}
Given nonzero unit base fees $\B p \in \R_{>0}^d$ and no resource-free transactions (i.e. $\B 0 \not \in T$),
\[\mathrm{INC} \le \mathrm{FEE} \le \mathrm{ORA} \le \mathrm{OPT}\]
\end{proposition}
\begin{proof}
    We first show that $F_{\B p} \subseteq I_{\B p}$. To see this, take any $\alpha \in F_{\B p}$. Suppose for contradiction that some $\beta \in V_{\B p}$ exists such that $T(\beta) \supset T(\alpha)$. It follows that some nonzero transaction $t_*$ lies in $T(\beta) \setminus T(\alpha)$. As unit base fees $\B p$ are nonzero across each dimension, 
    \[\sum_{t \in T(\beta)} \B g(t)^\top\B p - \sum_{t \in T(\alpha)} \B g(t)^\top\B p \ge \B g(t_*)^\top \B p > 0\]
    whence $\alpha$ is not fee-maximal as desired. It follows immediately from this fact that $\mathrm{INC} \le \mathrm{FEE}$. It can further be seen by inspection of the quantifiers in the maximums taken in \cref{eq:optrev}, \cref{eq:optora}, and \cref{eq:opt} that $\mathrm{FEE} \le \mathrm{ORA} \le \mathrm{OPT}$. 
\end{proof}

\subsection{Failure of Inclusion- and Fee-Maximal Allocations in the Ethereum Setting}

In this subsection we show that the surplus attained by inclusion- and fee- maximal allocations can suffer in a setting that models the Ethereum blockchain, even under best-case pricing. Formally, we consider a setting where there is a single node $N := \set{n}$. We endow this node $n$ with a \textit{capacity constraint} vector $\B r \in \Rg^d$ that indicates the maximum resource usage allowed along each dimension, and let the set of valid allocations be given by
\begin{equation}
    V := \set{\alpha \mid \sum_{t \in \alpha^{-1}(n)} g(t)_i \le r_i \text{ for all }i \in [d]}
\end{equation}
Under EIP-4844, for example, we would have one node (corresponding to Ethereum) that has resource constraints for gas as well as blobs (short-term storage). In our first result, we show that even if the base fee is set optimally, packing transactions arbitrarily into a block until the block is full can lead to poor surplus relative to $\mathrm{OPT}$ as the number of dimensions 
$d$ in the resource market increases. 

\begin{theorem}\label{thm:of}
    There exists a $d$-dimensional resource market with only one node such that 
    \[\mathrm{INC} \le \frac{2}{d} \cdot \mathrm{OPT}\]
\end{theorem}
\begin{proof}
For simplicity, we prove the statement when $c_n(\cdot) = 0$ (\textit{i.e.} the node is indifferent to any allocations so long as the capacity constraint is maintained), though we could in principle continue with any arbitrary cost function; to do so, it suffices to scale all of the valuations $v_{t}$ that we specify in this proof such that for all $t \in T$,
\begin{equation}
    v_{t} \gg \max \set{c_n(X) \mid X \subseteq T}
\end{equation}
so that the node's cost plays a negligible role in overall surplus. 

For all $i \in [d]$, let $r_i := 1$. We define the set of $d$ transactions $T$ and their valuations $v_{t}$ as follows. For each $j \in [d-1]$ we add a transaction $t^j$ to $T$, where
\begin{equation}
    g(t^j)_i := \begin{cases}
    1 & i = j \text{ or } i = d\\
    0 & \text{otherwise}
    \end{cases}
\end{equation}
For all $j \in [d-1]$, we let $v_{t^j} := \frac{2}{d}$. Next, we add a transaction $t^d$ to $T$, where  $g(t^d)_i := 1$ for all $i \in [d]$. We let $v_{t^d} := 1$. 
Noting that $t^d$ has the highest valuation (equal to $1$) among transactions and fits within the capacity of $n$, it follows that the maximum attainable surplus $\mathrm{OPT} \ge 1$ (the next sentence indeed implies that this inequality is in fact an equality).

Next, notice that because $r_d = 1$ and $g(t)_d = 1$ for all $t \in T$, it follows that for any valid allocation $\alpha \in V$, there is at most one transaction that can be allocated. Therefore any allocation $\alpha \in V_{\B p}$ under some given vector of unit base fees $\B p$ that allocates exactly one transaction is inclusion-maximal (\textit{i.e.} it is also in $I_{\B p}$). 

Let $\B p \in \Rg^d$ be an arbitrary vector of unit prices. Suppose first that $\B p^\top \B g(t^d) > 1 = v_{t^d}$. In this case, the surplus maximizing allocation (\emph{i.e.} the allocation that allocates $t^d$) is not in $V_{\B p} \supseteq I_{\B p}$. As every other allocation in $V$ yields surplus at most $\frac{2}{d}$, it follows that $\mathrm{INC} \le \frac{2}{d} = \frac{2}{d} \cdot \mathrm{OPT}$ as desired. Next, consider the case where $\B p^\top \B g(t^d) \le 1$. It follows that there exists some index $j \in [d-1]$ such that $p_j + p_d \le \frac{2}{d}$. However, as
\begin{equation}
    \B p^\top \B g(t^j) = p_j + p_d \le \frac{2}{d} = v_{t^j}
\end{equation}
the transaction $t^j$ is also willing to pay under unit prices $\B p$. By our observation earlier, the allocation that allocates just $t^j$ is inclusion-maximal. It follows that again an allocation exists in $I_{\B p}$ that obtains only $\frac{2}{d}$ surplus, whence we again have that $\mathrm{INC} \le \frac{2}{d} \cdot \mathrm{OPT}$ as desired.
\end{proof}

Next, we show that even if the computationally complex problem of fee-maximization were solved and unit base fees were set optimally, multi-dimensional fee markets cannot guarantee a surplus of more than $\frac12$ of optimal, even in the single-dimensional EIP-1559 setting. 

\begin{theorem}\label{thm:fee}
    For any $\epsilon > 0$, there exists a single-dimensional resource market with only one node such that 
    \[\mathrm{FEE} \le \frac{\mathrm{OPT}}{2} + \epsilon\]
\end{theorem}

\begin{proof}
    Again, for simplicity, we prove the statement when $c_n(\cdot) = 0$. We construct a single-dimensional example where there is a resource constraint $r_1 = k$. We let the set of transactions $T$ be given by $k + 1$ transactions $t^1, \dots, t^{k+1}$. For all $i \in [k + 1]$ (\textit{i.e.} for all transactions), we let the resource usage for the resource $g(t^i)_1 := 1$. This means valid allocations are ones that allocate at most $k$ of the $k + 1$ transactions to the node $n$. For $i \in [k]$ (\textit{i.e.} for $k$ of the transactions), we let the valuation of the transaction $v_{t^i} := \frac{1}{2(k-1)}$. For the last transaction, we let $v_{t^{k+1}} = \frac12$. It follows that the surplus-maximizing allocation includes the valuable transaction $t^{k+1}$ and $k - 1$ of the other transactions, whence we find that
    \[\mathrm{OPT} = \frac12 + (k-1) \cdot \frac{1}{2(k-1)} = 1\]
    Next, let $p$ be any unit base fee. If $p > 1$, then only the empty allocation lies in $V_p$, and all fee-maximal allocations attain $0$ surplus. If $\frac{1}{2(k-1)} < p < 1$, then the only allocation in $V_p$ is the allocation that only allocates $t^{k+1}$. In this case, that allocation is also the only fee-maximal allocation, and all fee-maximal allocations attain surplus $\frac12$. Finally, if $0 \le p \le \frac{1}{2(k-1)}$, then any allocation that allocates exactly $k$ of the $k+1$ transactions is both valid and maximizes fees (\textit{i.e.} lies in $R_p$), attaining total fees of $kp$. As the allocation that allocates $t^1, \dots, t^k$ lies in this set as well, and the surplus attained from this allocation is given by $k \cdot \frac{1}{2(k-1)}$, it follows from this (and the previous cases), that
    \[\mathrm{FEE} \le \frac{k}{2(k-1)}\]
    As we may take $k$ to be as large as we'd like, the desired result follows.
\end{proof}
\subsection{Failure of Oracle-Maximal Allocations in the Heterogenous Setting}

In this subsection, we consider a setting a where there are now multiple nodes with heterogeneous cost functions. To strengthen this impossibility result, we restrict our attention to nodes with cost functions that are linear in resource use. This aligns with the settings studied in \cite{MD-for-prover-markets} and \cite{lagrange-paper}. That is for any $n \in N$ and subset of transactions $B \subseteq T$, we let
\begin{equation}
    c_n(B) := \sum_{t \in B}\B g(t)^\top\B c^n
\end{equation}
where $\B c^n \in \Rg^d$ gives the unit resource costs for node $n$ along each dimension. We also endow each node $n$ with a capacity constraint $\B r^n \in \Rg^d$ that indicates maximum resource usage allowed along each dimension. The set of valid allocations in this multiple-node $d$-dimensional resource market is then given by
\begin{equation}
    V := \set{\alpha \mid \sum_{t \in \alpha^{-1}(n)} g(t)_i \le r^n_i \text{ for all }i \in [d]\text{ and }n \in N}
\end{equation}
We then study $d$-dimensional fee markets in the multiple-node case as we did before, with Equations \ref{eq:vp}, \ref{eq:pip}, and \ref{eq:phip}. The main result we show in this subsection is that if there are multiple nodes with heterogeneous cost functions, and a multidimensional resource market where multiple transactions can take non-trivial resource consumption amounts across some common dimension (\emph{i.e.} we haven't specified a unique dimension for each transaction), then even with unit base fees set optimally and an oracle-maximal allocation, multi-dimensional fee markets may obtain arbitrarily poor surplus relative to the best valid allocation. 

\begin{theorem}\label{thm:wo}
    For any $d \ge 1$, suppose that $k$ transactions in a $d$-dimensional resource market have different resource consumption values along some dimension, i.e. that for some $i \in [d]$
    \begin{equation}
        \abs{\set{g(t)_i \mid g(t)_i\not=0, t \in T}} = k
    \end{equation}
    Then, there exists valuations $v_{t}$ for each transaction in $T$, a set of nodes $N$, and cost vectors $\B c^n$ for each $n \in N$, such that 
    \[\mathrm{ORA} \le \frac{\mathrm{OPT}}{k}\]
\end{theorem}
\begin{proof}
    Define $t^1, \dots, t^k \in T$ such that
    \begin{equation}
        g(t^1)_i < g(t^2)_i < \dots < g(t^k)_i
    \end{equation}
    We now let there be $k$ nodes $N := \set{n^1, \dots, n^k}$. For each $j \in [k]$ and $\ell \in [d]$, we let 
    \begin{equation}
        r_{\ell}^{n^j} := \begin{cases}g(t^j)_i & \ell = i\\ \infty & \text{else}\end{cases}
    \end{equation}
    We then define for any $\ell \in [d]$,
    \begin{equation}
        c^{n^1}_\ell := \begin{cases}1 & \ell = i\\ 0 & \text{else}\end{cases}, \qquad \qquad v_{t^1} := g(t^1)_i + 1 
    \end{equation}
    and similarly, letting $\epsilon > 0$ be arbitrary, we define inductively for any $j \in \set{2, \dots, k}$ and $\ell \in [d]$ 
    \begin{equation}
        c^{n^j}_\ell := \begin{cases} v_{t^{j-1}}/g(t^{j-1})_i + \epsilon & \ell = i\\ 0 & \text{else}\end{cases}, \qquad \qquad v_{t^j} := c_i^{n^j} \cdot g(t^j)_i + 1 
    \end{equation}
    To see that the desired result follows from this construction, notice first that a valid allocation exists that yields $k$ surplus. This follows from the fact that we may allocate transaction $t^j$ to node $n^j$ for each $j \in [k]$. This allocation respects the capacity constraints, and the surplus generated by the allocation is given by 
    \begin{equation}
        \sum_{j = 1}^k v_{t^j} - c_i^{n^j} \cdot g(t^j)_i = k = \mathrm{OPT}
    \end{equation}
    Next, take any valid $d$-dimensional fee market allocation $\alpha \in V_{\B p}$ and suppose it maximizes surplus. For any $j \in [k]$, $t^j$ cannot be allocated to $n^{\ell}$ for any $\ell < k$ as $r^{n^\ell}_i < g(t^j)_i$ by construction. Further, $t^j$ cannot be allocated to $n^{\ell}$ for any $\ell > k$ because for any such $\ell$, 
    \begin{equation}
        v_{t^j} < c_i^{n^\ell} \cdot g(t^j)_i
    \end{equation}
    whence any allocation that allocates $t^j$ to $n^{\ell}$ would reduce surplus relative to the allocation that leaves $t^j$ unallocated. It follows that $\alpha$ can only allocate some subset of the $t^j$ to their corresponding nodes $n^j$. Such an allocation can only exist in $V_{\B p}$ and not reduce surplus relative to leaving $t^j$ unallocated if 
    \begin{equation}
        c_i^{n^j} \le p_i \le v^{t^j}/g(t^j)_i
    \end{equation}
    However, because by construction, all of the intervals $[c_i^{n^j}, v_{t^j}/g(t^j)_i]$ for $j \in [k]$ are disjoint from each other, this may hold for at most one value of $j$. It follows that the total surplus attained by any such $\alpha \in V_{\B p}$ is $1$, whence $\mathrm{ORA} = 1$ and the desired result holds.
\end{proof}

As a consequence, when nodes have heterogeneous costs, it is necessary to set individual prices for transactions and nodes. This is a key property of the Resonance mechanism.

\end{document}